\documentclass[]{elsarticle}
\usepackage{hyperref}
\usepackage{array,graphicx,makeidx,amssymb,amsmath,booktabs,comment,soul,color,multirow,tkz-graph}
\usepackage[lined,ruled,noend]{algorithm2e}
\usepackage[utf8]{inputenc}
\usepackage[symbol]{footmisc}
\usepackage{amsmath}
\usepackage{amssymb}
\usepackage{amsthm}

\definecolor{light-gray}{gray}{0.87}

\newenvironment{malgorithm}{\begin{algorithm}\footnotesize\setlength{\parskip}{0.2pt}}{\end{algorithm}}

\newtheorem{theorem}{Theorem}[section]

\newtheorem{lemma}[theorem]{Lemma}

\newtheorem{remark}[theorem]{Remark}

\theoremstyle{definition}
\newtheorem{definition}[theorem]{Definition}

\journal{Journal of \LaTeX\ Templates}

\begin{document}

\begin{frontmatter}

\title{Order-Preserving Pattern Matching Indeterminate Strings}

\author[inesc]{Diogo Costa}
\author[inesc]{Luís M. S. Russo}
\author[inesc]{Rui Henriques}
\author[kyushu]{Hideo Bannai}
\author[inesc]{Alexandre P. Francisco}\cortext[mycorrespondingauthor]{Corresponding author}\ead{aplf@tecnico.ulisboa.pt}

\address[inesc]{INESC-ID and Instituto Superior Técnico, Universidade de Lisboa, Portugal}
\address[kyushu]{Department of Computer Science, Kyushu University, Japan}

\begin{abstract} 
\noindent Given an indeterminate string pattern $p$ and an indeterminate string text $t$, the problem of order-preserving pattern matching with character uncertainties ($\mu$OPPM) is to find all substrings of $t$ that satisfy one of the possible orderings defined by $p$. When the text and pattern are determinate strings, we are in the presence of the well-studied exact order-preserving pattern matching (OPPM) problem with diverse applications on time series analysis. Despite its relevance, the exact OPPM problem suffers from two major drawbacks: 1) the inability to deal with indetermination in the text, thus preventing the analysis of noisy time series; and 2) the inability to deal with indetermination in the pattern, thus imposing the strict satisfaction of the orders among all pattern positions. %

This paper provides the first polynomial algorithm to answer the $\mu$OPPM problem when indetermination is observed on the pattern or text. Given two strings with length $m$ and $O(r)$ uncertain characters per string position, we show that the $\mu$OPPM problem can be solved in $O(mr\lg r)$ time when one string is indeterminate and $r\in\mathbb{N}^+$. %
 Mappings into satisfiability problems are provided when indetermination is observed on both the pattern and the text, and results concerning the general problem complexity are presented as well, with $\mu$OPPM problem proved to be \textbf{NP}-hard in general.
\end{abstract}


\end{frontmatter}



%
%
%
%
%
%

%
\section{Introduction}
Given a pattern string $p$ and a text string $t$, the exact order preserving pattern matching (OPPM) problem is to find all substrings of $t$ with the same relative orders as $p$. The problem is applicable to strings with characters drawn from numeric or ordinal alphabets. Illustrating, given $p$=(1,5,3,3) and $t=(5,1,4,2,2,5,2,4)$, substring $t[1..4]=(1,4,2,2)$ is reported since it satisfies the character orders in $p$, $p[0]\le p[2]=p[3]\le p[1]$. %
Despite its relevance, the OPPM problem has limited potential since it prevents the specification of errors, uncertainties or don't care characters within the text. %

Indeterminate strings allow uncertainties between two or more characters per position. Given indeterminate strings $p$ and $t$, the problem of order preserving pattern matching uncertain text ($\mu$OPPM) is to find all substrings of $t$ with an assignment of values that satisfy the orders defined by $p$. For instance, let $p=(1,2|5,3,3)$ and $t=(5,0,1,2|1,2,5,2|3,3|4)$. The substrings $t[1..4]$ and $t[4..7]$ are reported since there is an assignment of values that preserve either $p[0]<p[1]<p[2]=p[3]$ or $p[0]<p[2]=p[3]<p[1]$ orderings: respectively $t[1..4]=(0,1,2,2)$ and $t[4..7]=(2,5,3,3)$. %

Order-preserving pattern matching captures the structural isomorphism of strings, therefore having a wide-range of relevant applications in the analysis of financial times series, musical sheets, physiological signals and biological sequences \cite{ge1998pattern,kim,henriques2014seven}. Uncertainties often occur across these domains. In this context, although the OPPM problem is already a relaxation of the traditional pattern matching problem, the need to further handle localized errors is essential to deal with noisy strings \cite{ruiphd}. For instance, given the stochasticity of gene regulation (or markets), the discovery of order-preserving patterns in gene expression (or financial) time series needs to account for uncertainties \cite{bicspam,ruispm}. Numerical indexes of amino-acids (representing physiochemical and biochemical properties) are subjected to errors difficulting the analysis of protein sequences \cite{kawashima2000aaindex}. Another example are ordinal strings obtained from the discretization of numerical strings, often having two uncertain characters in positions where the original values are near a discretization boundary \cite{ruiphd}. %

Let $m$ and $n$ be the length of the pattern $p$ and text $t$, respectively. The exact OPPM problem has a linear solution on the text length $O(n+m \lg m)$ based on the Knuth-Morris-Pratt algorithm \cite{kubica,kim,cho1}. Alternative algorithms for the OPPM problem have also been proposed \cite{cho2,belazzougui,Chhabra}. %
Contrasting with the large attention given to the resolution of the OPPM problem, to our knowledge there are no polynomial-time algorithms to solve the $\mu$OPPM problem. Naive algorithms for $\mu$OPPM assess all possible pattern and text assignments, bounded by $O(n r^{m})$ when considering up to $r$ uncertain characters per position.

This work proposes the first polynomial time algorithms able to answer the $\mu$OPPM problem. Accordingly, the contributions are organized as follows. First, we show that an indeterminate string of length $m$ order-preserving matches a determinate string with the same length in $O(mr\lg r)$ time based on their monotonic properties.
Second, and given two indeterminate strings with the same size, we provide a linear encoding of the $\mu$OPPM into a satisfiability formula with properties of interest.
Furthermore, we extend this encoding and we present results concerning the computational complexity of $\mu$OPPM problem variations, namely a proof of that the $\mu$OPPM problem is \textbf{NP}-hard in general.
Third, given a pattern and text strings with lengths $m$ and $n$, only one of them indeterminate, we show that the $\mu$OPPM problem can be solved in linear space and its average efficiency boosted under effective filtration procedures.

A preliminary version of this work was presented at the Annual Symposium on Combinatorial Pattern Matching (CPM)~\cite{cpm2018}.
In this paper, we revise previous results and we present new results concerning the computational complexity of $\mu$OPPM problem; Sections~\ref{sec:alternateoppm}, \ref{sec:reduction} and \ref{sec:open_problem} are new.

\section{Background}
Let $\Sigma$ be a totally ordered alphabet and an element of $\Sigma^{\ast}$ be a string. The length of a string $w$ is denoted by $|w|$. The empty string $\varepsilon$ is a string of length 0. For a string $w=xyz$, $x$, $y$ and $z$ are called a prefix, substring, and suffix of $w$, respectively. The $i$-th character of a string $w$ is denoted by $w[i]$ for each $0\le i< |w|$. For a string $w$ and integers $0\le i\le j< |w|$, $w[i..j]$ denotes the substring of $w$ from position $i$ to position $j$. For convenience, let $w[i..j]=\varepsilon$ when $i>j$. 

Given strings $x$ and $y$ with equal length $m$, $y$ is said to order-preserving against $x$ \cite{kubica}, denoted by $x\approx y$, if the orders between the characters of $x$ and $y$ are the same, i.e. $x[i]\le x[j] \Leftrightarrow y[i]\le y[j]$ for any $0\le i,j< m$. 
A non-empty pattern string $p$ is said to order-preserving match (\textit{op-match} in short) a non-empty text string $t$ if and only if there is a position $i$ in $t$ such that $p\approx t[i-|p|+1..i]$. The \textit{order-preserving pattern matching} (OPPM) problem is to find all such text positions.

\subsection{The Problem}
Given a totally ordered alphabet $\Sigma$, an indeterminate string is a sequence of disjunctive sets of characters $x[0]x[1]..x[n-1]$ where $x[i]\subseteq\Sigma$. Each position is given by $x[i]=\sigma_1..\sigma_r$ where $r\ge 1\wedge \sigma_i\in\Sigma$. %

Given an indeterminate string $x$, a \textit{valid assignment} $\$x$ is a (determinate) string with a single character at position $i$, denoted $\$x[i]$, contained in the $x[i]$ set of characters, i.e. $\$x[0]\in x[0]$, \ldots, $\$x[m-1]\in x[m-1]$. For instance, the indeterminate string $(1|3,3|4,2|3,1|2)$ has $2^4$ valid assignments. Given an indeterminate position $x[i]\subseteq\Sigma$, $\$x_j[i]$ is the $j^{th}$ ordered value of $x[i]$ (e.g. $\$x_0[i]$=1 for $x[i]=1|2$). Given an indeterminate string $x$, let a \textit{partially assigned} string $\S x$ be an indeterminate string with an arbitrary number of uncertain characters removed, i.e. $\S x[0]\subseteq x[0]$, \ldots, $\S x[m-1]\subseteq x[m-1]$. 

Given a determinate string $x$ of length $m$, an indeterminate string $y$ of equal length is said to be \textit{order-preserving} against $x$, identically denoted by $x\approx y$, if there is a valid assignment $\$y$ such that the relative orders of the characters in $x$ and $\$y$ are the same, i.e. $x[i]\le x[j] \Leftrightarrow \$y[i]\le 
\$y[j]$ for any $0\le i,j< m$. Given two indeterminate strings $x$ and $y$ with length $m$, $y$ preserves the orders of $x$, $x\approx y$, if exists $\$y$ in $y$ that respects the orders of a valid assignment $\$x$ in $x$.

A non-empty indeterminate pattern string $p$ is said to order-preserving match (\textit{op-match} in short) a non-empty indeterminate text string $t$ if and only if there is a position $i$ in $t$ such that $p\approx t[i-|p|+1..i]$. The problem of \textit{order-preserving pattern matching with character uncertainties} ($\mu$OPPM) problem is to find all such text positions.

To understand the complexity of the $\mu$OPPM problem, let us look to its solution from a naive stance yet considering state-of-the-art OPPM principles. The algorithmic proposal by Kubica et al.~\cite{kubica} is still up to this date the one providing a lowest bound, $O(n$+$q)$, where $q=m$ for alphabets of size $m^{O(1)}$ ($q=m\lg m$ otherwise). %
Given a determinate string $x$ of length $m$, an integer $i$ ($0\le i< m$) is said in the context of this work to be an \textit{order-preserving border} of $x$ if $x[0..i]\approx x[m-i+1..m]$. 
In this context, given a pattern string $p$, the orders between the characters of $p$ are used to linearly infer the order borders. The order borders can then be used within the Knuth-Morris-Pratt algorithm to find op-matches against a text string $t$ in linear time \cite{kubica}. 

Given a determinate string $p$  of length $m$ and an indeterminate string $t$ of length $n$, the previous approach is a direct candidate to the $\mu$OPPM problem by decomposing $t$ in all its possible assignments, $O(r^n)$. %
Since determinate assignments to $t$ are only relevant in the context of $m$-length windows, this approach can be improved to guarantee a maximum of $O(r^m)$ assignments at each text position. %
Despite its simplicity, this solution is bounded by $O(nr^m)$. This complexity is further increased when indetermination is also considered in the pattern, stressing the need for more efficient alternatives.\vskip -0.75cm\textcolor{white}{.}%

\subsection{Related work} 

The exact OPPM problem is well-studied in literature. Kubica et al. \cite{kubica}, Kim et al. \cite{kim} and Cho et al. \cite{cho1} presented linear time solutions on the text length by respectively combining order-borders, rank-based prefixes and grammars with the Knuth–Morris–Pratt (KMP) algorithm \cite{kpm}. Cho et al. \cite{cho2}, Belazzougui et al. \cite{belazzougui}, and Chhabra et al. \cite{Chhabra} presented $O(nm)$ algorithms that show a sublinear average complexity by either combining bad character heuristics with the Boyer–Moore algorithm \cite{boyermoore} or applying filtration strategies. Recently, Chhabra et al. \cite{Chhabra:2017:EOP:3129357.3129363} proposed further principles to solve OPPM using word-size packed string matching instructions to enhance efficiency.

In the context of numeric strings, multiple relaxations to the exact pattern matching problem have been pursued to guarantee that approximate matches are retrieved. %
In norm matching \cite{nm1,nm2,nm3,nm4}, matches between numeric strings occur if a given distance threshold $f(x,y)\le \theta$ is satisfied. In ($\delta$,$\gamma$)-matching \cite{dg1,dg2,dg3,dg4,dg5,dg6,dg7}, %
strings are matched if the maximum difference of the corresponding characters is at most $\delta$ and the sum of differences is at most $\gamma$. 

In the context of nominal strings, variants of the pattern matching task have also been extensively studied to allow for don't care symbols in the pattern \cite{indeterminate,r17,gpm}, transposition-invariant \cite{dg5}, parameterized matching \cite{pm1,pm2}, less than matching \cite{r6}, swapped matching \cite{sm1,sm2}, gaps \cite{gap1,gap2,gap3}, overlap matching \cite{r5}, and function matching \cite{fm1,fm2}. 

Despite the relevance of the aforementioned contributions to answer the exact order-preserving pattern matching and generic pattern matching, they cannot be straightforwardly extended to efficiently answer the $\mu$OPPM problem.

\section{On solving $\mu$OPPM}
\label{secSol1}
Section~\ref{uOPPM1} introduces the first efficient algorithm to solve the $\mu$OPPM problem when one string is indeterminate ($r\in\mathbb{N}^+$).
Section~\ref{2indsec} discusses the existence of efficient solvers when both strings are indeterminate. %
Section~\ref{sec:alternateoppm} introduces then a polynomial time algorithm for the Alternate-$\mu$OPPM as a subproblem of $\mu$OPPM where both strings may have indeterminate characters, but never in the same position.
Given the formulations proposed in Section~\ref{2indsec}, we hypothesize that op-matching indeterminate strings with an arbitrary number of uncertain characters per position ($r\in\mathbb{N}^+$) is in class \textbf{NPC}. 
Furthermore, we show in Section~\ref{sec:reduction} that the problem \{3,3\}-$\mu$OPPM, defined as the subproblem of $\mu$OPPM where both the pattern and the text have indeterminate characters in any position (although at least one position must have at least three indeterminate characters in both pattern and text), is \textbf{NP}-hard.
We still leave a gap in between these two groups, namely for the strings where there are \textit{at most two} indeterminate characters in both strings at the same position. It remains open whether or not this problem is \textbf{NP}-hard. 

\subsection{\textbf{$O(mr\lg r)$} time $\mu$OPPM when one string is indeterminate}
\label{uOPPM1}

Given a determinate string $x$ of length $m$, there is a well-defined permutation of positions, $\pi$, that specifies a non-monotonic ascending order of characters in $x$. %
For instance, given $x$=(1,4,3,1), then $x[0]=x[3] < x[2] < x[1]$ and $\pi=(0,3,2,1)$. %
Given a determinate string $y$ with the same length, $y$ op-matches $x$ if it $y$ satisfies the same $m$-1 orders. For instance, given $x=(1,4,3,1)$ and $y=(2,5,4,3)$, $x$ orders are not preserved in $y$ since $y[0]\mathbin{\textcolor{red}{\neq}}y[3]<y[2]<y[1]$.

The monotonic properties can be used to answer $\mu$OPPM when one string is indeterminate. Given an indeterminate string $y$, let $x_\pi$ and $y_\pi$ be the permuted strings in accordance with $\pi$ orders in $x$. To handle equality constraints, positions in $y_\pi$ with identical characters in $x_\pi$ can be intersected, producing a new string $y'_{\pi}$ with $s$ length ($s\le m$). Illustrating, given $x$=(4,1,4,2) and $y=(2|7, 2,7|8,1|4|8)$, then ${\pi}$=(1,3,0,2), $x_{\pi}$=(1,2,4,4), $y_{\pi}=(2,8|4|1,7|2,8|7)$ and $y'_{\pi}=(y_{\pi}[0],y_{\pi}[1],$ $y_{\pi}[2]\cap y_{\pi}[3])=(2,8|4|1,7)$. To handle monotonic inequalities, $y'_\pi[i]$ characters can be concatenated in descending order to compose $z=y'_\pi[0]y'_\pi[1]..y'_\pi[s]$ and the orders between $x$ and $y$ verified by testing if the longest increasing subsequence (LIS)~\cite{FREDMAN197529} of $z$ has $s$ length. In the given example, $z=(2,8,4,1,7)$, and the LIS of $z=(\textbf{2},8,\textbf{4},1,\textbf{7})$ is $w$=(2,4,7). Since $|w|=|y'_\pi|$=3, $y$ op-matches $x$. %

\begin{theorem}
Given a determinate string $x$ and an indeterminate string $y$, let $x_\pi$ and $y_\pi$ be the sorted strings in accordance with $\pi$ order of characters in $x$. Let the positions with equal characters in $x_\pi$ be intersected in $y_\pi$ to produce a new indeterminate string $y'_\pi$. Consider $z_i$ to be a string with $y'_\pi[i]$ characters in descending order and $z=z_1z_2..z_m$, then $|w|=|y'_\pi|$ if and only if $y\approx x$, where $w$ is a longest increasing subsequence in $z$.
\label{T1}	
\end{theorem}
\begin{proof}
\textbf{$(\Rightarrow)$} If the length of the longest increasing subsequence (LIS), $|w|$, equals the number of monotonic relations in $x$, $|y'_\pi|$, then $y\approx x$. By sorting characters in descending order per position, we guarantee that at most one character per position in $y'_\pi$ appears in the LIS (respecting monotonic orders in $x$ given $y'_\pi$ properties). By intersecting characters in positions of $y$ with identical characters in $x$, we guarantee the eligibility of characters satisfying equality orders in $x$, otherwise empty positions in $y'_\pi$ are observed and the LIS length is less than $|y'_\pi|$. %
\textbf{$(\Leftarrow)$} If $|w|<|y'_\pi|$, there is no assignment in $y$ that op-matches $x$ due to one of two reasons: 1) there are empty positions in $y'_\pi$ due to the inability to satisfy equalities in $x$, or 2) it is not possible to find a monotonically increasing assignment to $y'_\pi$ and, given the properties of $y'_\pi$, $y_\pi$ cannot preserve the orders of $x_\pi$.
\end{proof}
Solving the LIS task on a string of size $n$ is $O(n\lg n)$~\cite{FREDMAN197529} where $n=|z|=O(rm)$. In addition, set intersection operations are performed $O(m)$ times on sets with $O(r)$ size, which can be accomplished in $O(rm\lg r)$ time. As a result, the $\mu$OPPM problem with one indeterminate string can be solved in $O(rm\lg(rm))$. %

Given the fact that the candidate string for the LIS task has properties of interest, we can improve the complexity of this calculus (Theorem~\ref{timeLIS}) in accordance with Algorithm~\ref{algo1}.%
\begin{malgorithm}
\KwIn{determinate $x$, indeterminate $y$ ($|x|=|y|=m$)}
$\pi$ $\leftarrow$ sortedIndexes($x$)\tcp*{$O(m)$ if $|\Sigma|=m^{O(1)}$; $O(m\lg m)$ otherwise}
$x_\pi$ $\leftarrow$ permute($x$,$\pi$), $y_\pi$ $\leftarrow$ permute($y$,$\pi$)\tcp*{$O(m+mr)$}
$j$ $\leftarrow$ 0; $y'_\pi[0]$ $\leftarrow$ $\{y_\pi[0]\}$\;
\ForEach(\tcp*[f]{$O(mr\lg r)$}){$i$ $\in$ 1..$m$-1}{
   \lIf(\tcp*[f]{$O(r\lg r)$}){$x_\pi[i]=x_\pi[i$-$1]$}{$y'_\pi[j]$ $\leftarrow$ $y'_\pi[j]\cap \{y_\pi[i]\}$}
   \lElse{$j$ $\leftarrow$ $j$+1; $y'_\pi[j]$ $\leftarrow$ $\{y_\pi[i]\}$}
}
$s$ $\leftarrow$ $|y'_\pi|$, nextMin $\leftarrow$ -$\infty$\;
\ForEach(\tcp*[f]{$O(m r)$}){$i$ $\in$ 0..$s$-$1$}{
     nextMin $\leftarrow$ min$\{ a \mid a\in y'_\pi[i], a$$>$nextMin$\}$;\tcp*[f]{$O(r)$}\\
     \lIf{$\not\exists$ \emph{nextMin}}{\textbf{return} \textsf{false}}
}
\textbf{return} \textsf{true}\;
\caption{$O(mr\lg r)$ $\mu$OPPM algorithm with one indeterminate string.\label{algo1}}
\end{malgorithm}

\begin{theorem}
$\mu$OPPM two strings of length $m$, one being indeterminate, is in $O(mr\lg r)$ time, where $r\in\mathbb{N}+$. \label{timeLIS} %
\end{theorem}
\begin{proof} In accordance with Algorithm \ref{algo1}, $\mu$OPPM is bounded by the verification of equalities, $O(mr\lg r)$ \cite{demaine2000adaptive}. 
	Testing inequalities after set intersections can be linearly performed on the size of $y$, $O(mr)$ time, improving the $O(mr\lg(mr))$ bound given by the LIS calculus. %
\end{proof}

The analysis of Algorthim \ref{algo1} further reveals that the $\mu$OPPM problem with one indeterminate string requires linear space in the text length, $O(mr)$.

\subsection{$\mu$OPPM with indeterminate pattern and text}
\label{2indsec}

As indetermination in real-world strings is typically observed between pairs of characters \cite{ruiphd}, a key question is whether $\mu$OPPM on two indeterminate strings is in class \textbf{P} when $r=2$. %
To explore this possibility, new concepts need to be introduced. In OPPM research, character orders in a determinate string of length $m$ can be decomposed in 3 sequences with $m$ unit sets: 

%
\begin{definition} For $i=0,\ldots,m-1$:
	\begin{itemize}
		\item $\textit{Leq}_x[i]=\{\max\{ k \mid k < i, x[i]=x[k] \}\}$ ($\emptyset$ if there is no eligible $k$),
		\item $\textit{Lmax}_x[i]=\{\max\{\mathrm{argmax}_k\{ x[k] \mid k<i, x[i]>x[k] \}\}\}$ ($\emptyset$ if there is no eligible $k$),
		\item $\textit{Lmin}_x[i]=\{\max\{\mathrm{argmin}_k\{ x[k] \mid k<i, x[i]<x[k] \}\}\}$ ($\emptyset$ if there is no eligible $k$).
	\end{itemize} 
	\label{minmaxoriginal}
\end{definition}

\textit{Leq}, \textit{Lmax} and \textit{Lmin} capture $=$, $>$ and $<$ relationships between each character $x[i]$ in $x$ and the closest preceding character $x[k]$. %
These orders can be inferred in linear time for alphabets of size $m^{O(1)}$ and in $O(m \lg m)$ time for other alphabets by answering the ``all nearest smaller values'' task on the sorted indexes \cite{kubica}. Figure~\ref{simplekubica} depicts \textit{Leq}, \textit{Lmax} and \textit{Lmin} for $x=(1,4,3,1)$. %
Given determinate strings $x$ and $y$, $A=Leq_x[t+1]$, $B=Lmax_x[t+1]$ and $C=Lmin_x[t+1]$, if $x[0 . . t]\approx y[0 . . t]$, then $x[0 . . t + 1]\approx y[0 . . t + 1]$ if and only if
\begin{equation*}
\forall_{a\in A }\ ( y[t + 1] = y[a] ) \wedge \forall_{b\in B }\ ( y[t + 1] > y[b] ) \wedge \forall_{c\in C }\ ( y[t + 1] < y[c]).
\end{equation*}

\begin{figure}[t]
  \centering
	\begin{minipage}[b]{\linewidth}\centering
		\begin{tikzpicture}[xscale=1,yscale=1]
		\SetGraphUnit{1.5}
		\Vertex{1}\EA(1){4}\EA(4){3}\EA(3){1 }
		\tikzset{EdgeStyle/.append style={->, bend right = 40, color=purple}}
		\Edge[label = ](4)(1)
		\Edge[label = $<$](3)(1)
		\tikzset{EdgeStyle/.append style={->, bend left = 30, color=blue}}
		\Edge[label  = $>$ ](3)(4)
		\tikzset{EdgeStyle/.append style={->, bend left = 30, color=green}}
				\Edge[label = ${=}$ ](1 )(1)
		\tikzset{EdgeStyle/.append style = {bend left = 40, line width=1pt}}
		\end{tikzpicture}
	\end{minipage}\\
	\begin{minipage}[t]{\linewidth}\centering
		\begin{tabular}{lm{.35cm}m{.35cm}m{.35cm}m{.35cm}}\toprule
			Pattern&1&\hspace{0.1cm}4&\hspace{0.1cm}3&\hspace{0.1cm}1\\
			\midrule
			\mbox{\textit{Leq}$[i]$}&$\emptyset$&\hspace{0.1cm}$\emptyset$&\hspace{0.1cm}$\emptyset$&\{0\}\\
			Ordered indexes (asc)&0&\hspace{0.1cm}3&\hspace{0.1cm}2&\hspace{0.1cm}1\\
			\mbox{\textit{Lmax}$[i]$ (nearest asc smaller not in \textit{Leq}$[i]$)}&$\emptyset$&\{0\}&\{0\}&\hspace{0.1cm}$\emptyset$\\
			Ordered indexes (desc)&2&\hspace{0.1cm}0&\hspace{0.1cm}1&\hspace{0.1cm}3\\
			\mbox{\textit{Lmin}$[i]$ (nearest desc smaller not in \textit{Leq}$[i]$)}&$\emptyset$&\hspace{0.1cm}$\emptyset$&\{1\}&\hspace{0.1cm}$\emptyset$\\\bottomrule
		\end{tabular}
	\end{minipage}
	\caption{Orders identified for $p=(1,4,3,1)$ where \textit{Leq}, \textit{Lmax} and \textit{Lmin} are in accordance with Kubica et al.~\cite{kubica}.}
	\label{simplekubica}
\end{figure}

When allowing uncertainties between pairs of characters, previous research on the OPPM problem cannot be straightforwardly extended %
due to the need to trace $O(2^m)$ assignments on indeterminate strings. %

\begin{lemma}
	Given a determinate string $x$, an indeterminate string $y$, and the singleton sets $A=Leq_x[t+1]$, $B=Lmax_x[t+1]$ and $C=Lmin_x[t+1]$ containing a position in $\{0,\ldots,t\}$. If $x[0 . . t]\approx y[0 . . t]$ is verified on a specific assignment of $y$ characters, denoted $\S y$, then $x[0 . . t + 1]\approx y[0 . . t + 1]$ if and only if
\begin{equation*}
\begin{split}
	\exists_{\$y[t+1]\in\S y[t+1]}\ \forall_{a\in A}\ \exists_{\$y[a]\in\S y[a]}\ \forall_{b\in B}\ \exists_{\$y[b]\in\S y[b]}\ \forall_{c\in C}\ \exists_{\$y[c]\in\S y[c]} \\
	\$y[t+1]= \$y[a] \wedge \$y[t + 1] > \$y[b] \wedge \$y[t + 1] < \$y[c]
\end{split}
\end{equation*}
	\label{orderlemma}
\end{lemma}
\begin{proof}
	\textbf{$(\Rightarrow)$} In accordance with \textit{Leq}, \textit{Lmax} and \textit{Lmin} definition, for any $a\in A$, $b\in B$ and $c\in C$ we have $x[t + 1]=x[a]$, $x[t + 1]>x[b]$ and $x[t +
1]<x[c]$. If there is an assignment to $y[0 . . t + 1]$ in $\S y$ that preserves the orders of $x[0 . . t + 1]$, then for each $a\in A$, $b\in B$ and $c\in C$ $\$y[t + 1]=\$y[a]$, $\$y[t + 1]>\$y[b]$ and $\$y[t +
1]<\$y[c]$ (where $\$y[t + 1]\in \S y[t+1]$, $\$y[a]\in \S y[a]$, $\$y[b]\in \S y[b]$, $\$y[c]\in \S y[c]$). 
\textbf{$(\Leftarrow)$} We need to show that $x[0 . . t + 1]\approx y[0 . . t + 1]$. Since $x[0 . . t]\approx y[0 . . t]$, for $i < t$, $\exists_{\$y[i]\in \S y[i], \$y[t + 1]\in \S y[t+1]}$:
$x[t+1]>x[i] \Leftrightarrow \$y[t+1]>\$y[i]$.
Assuming $x[t+1]>x[i]$ for some $i \in \{0,\ldots,t\}$: by the definition of \textit{Lmax}, $\forall_{b\in B}x[b]>x[i]$; by the order-isomorphism of $x[0 . . t]$ and $\$y[0 . . t]$ in $\S y[0 . . t]$, there is $\$y[i]\in \S y[i]$ and $\$y[b]\in \S y[b]$
that $\forall_{b\in B}\$y[b]>\$y[i]$; and by the assumption of the lemma, $\forall_{b\in B}\$y[t+1]>\$y[b]$; hence
$\$y[t+1]>\$y[i]$. Similarly, $x[t+1]<x[i]$ (and $x[t+1]=x[i]$) implies  $\$y[t+1]<\$y[i]$ (and $\$y[t+1]=\$y[i]$), yielding the stated equivalence. %
\end{proof}
Given two strings of equal length, the $\mu$OPPM problem can be schematically represented according to the identified order restrictions. Figure~\ref{ordersFig} represents restrictions on the indeterminate string $y=(2,4|5,3|5,1|2)$ in accordance with the observed orders in $x=(1,4,3,1)$. %
The left side edges are placed in accordance with Lemma~\ref{orderlemma} and capture assessments on the orders between pairs of characters. The right side edges capture incompatibilities detected after the assessments, i.e. pairs of characters that cannot be selected simultaneously (for instance, $y[0]=2$ and $y[3]=1$, or $y[1]=4$ and $y[2]=5$). For the given example, there are two valid assignments,  $\$y_1=(2,4,3,2)$ and $\$y_2=(2,5,3,2)$, that satisfy $x[0]=x[3]<x[2]<x[1]$, thus $y$ op-matches $x$.

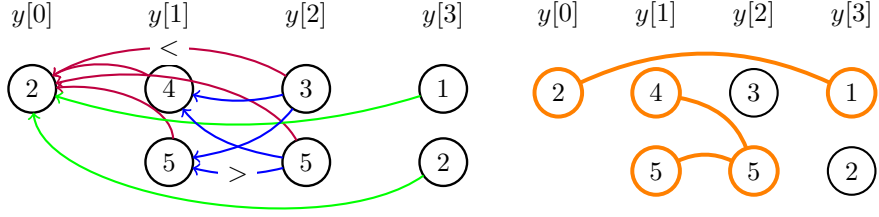
\begin{figure}[t]
	\begin{minipage}[b]{0.55\linewidth}\centering
		\begin{tikzpicture}[xscale=1.4,yscale=0.75]
		\SetGraphUnit{1.3}
		\SetVertexNormal[LineColor=white,LineWidth=2pt]
		\Vertex[L={$y[0]$}]{x1}\EA[L={$y[1]$}](x1){x2}\EA[L={$y[2]$}](x2){x3}\EA[L={$y[3]$}](x3){x4}
		\SetVertexNormal[LineColor=black,LineWidth=.9pt]
		\SO[L=2](x1){2a}
		\SO[L=4](x2){4a}\SO[L=5](4a){5a}
		\SO[L=3](x3){4b}\SO[L=5](4b){5b}
		\SO(x4){1}\SO[L=2](1){2b}
		\tikzset{EdgeStyle/.append style={->, bend right = 50, color=purple}}
		\Edge(4a)(2a)
		\Edge[label=$<$](4b)(2a)
		\Edge(5a)(2a)
		\Edge(5b)(2a)
		\tikzset{EdgeStyle/.append style={->, bend left = 70, color=green}}
		\Edge(2b)(2a)
		\tikzset{EdgeStyle/.append style={->, bend left = 30, color=green}}
		\Edge(1)(2a)
		\tikzset{EdgeStyle/.append style={->, bend left = 25, color=blue}}
		\Edge(4b)(4a)
		\Edge(4b)(5a)
		\Edge(5b)(4a)
		\Edge[label=$>$](5b)(5a)
		\end{tikzpicture}	
	\end{minipage}
	\quad
	\begin{minipage}[t]{0.35\linewidth}\centering\vskip -3.52cm
		\begin{tikzpicture}[xscale=1,yscale=0.8]
		\SetGraphUnit{1.3}
		\SetVertexNormal[LineColor=white,LineWidth=2pt]
		\Vertex[L={$y[0]$}]{x1}\EA[L={$y[1]$}](x1){x2}\EA[L={$y[2]$}](x2){x3}\EA[L={$y[3]$}](x3){x4}
		\SetVertexNormal[LineColor=orange,LineWidth=1.5pt]
		\SO[L=2](x1){2a}
		\SO(x4){1}
		\SO[L=4](x2){4a}
		\SO[L=5](4a){5a}
		\SetVertexNormal[LineColor=black,LineWidth=.8pt]		\SO[L=3](x3){4b}
		\SO[L=2](1){2b}
		\SetVertexNormal[LineColor=orange,LineWidth=1.5pt]
		\SO[L=5](4b){5b}
		\tikzset{EdgeStyle/.append style={-, bend right, color=orange, line width=1.5pt}}
		\Edge(1)(2a)
		\Edge(5b)(4a)
		\Edge(5b)(5a)
		\end{tikzpicture}
	\end{minipage}
	\vskip -0.5cm
	\caption{Schematic representation of the pairwise ordering restrictions for text $y$=($2,4|5,3|5,1|2$) and pattern $x$=(1,4,3,1). In the left side, all order verifications are represented, while in the right side only the order conflicts are signaled (e.g. $y[1]$=4 cannot be selected together with $y[2]$=5).}
	\label{ordersFig}
\end{figure}

To verify whether there is an assignment that satisfies the identified ordering restrictions, we propose the reduction of $\mu$OPPM problem to a Boolean satisfiability problem.

Given a set of Boolean variables, a formula in conjunctive normal form is a conjunction of clauses, where each clause is a disjunction of literals, and a literal corresponds to a variable or its negation. Let a 2CNF formula be a formula in the conjunctive normal form with at most two literals per clause. Given a CNF formula, the \textit{satisfiability} (SAT) problem is to verify if there is an assigning of values to the Boolean variables such that the CNF formula is satisfied.

\begin{theorem}
	The $\mu$OPPM problem over two strings of equal length, one being indeterminate, can be reduced to a satisfiability problem with the following CNF formula:%
\begin{equation}
\begin{split}
\phi= & \bigwedge_{i=0}^{m-1} \left(\bigvee_{\$y[i]\in y[i]} z_{i,\$y[i]}\right) \\
& \wedge  \bigwedge_{i=0}^{m-1} 
  \left(           \bigwedge_{\$y[i]\in y[i]}\bigwedge_{\substack{j\in  Leq[i]\\\$y[j]\in y[j]}} \left(\neg z_{i,\$y[i]} \vee \neg z_{j,\$y[j]} \vee \$y[i] = \$y[j] \right) \right.\\
& \qquad\ \ \wedge \bigwedge_{\$y[i]\in y[i]}\bigwedge_{\substack{j\in Lmax[i]\\\$y[j]\in y[j]}} \left(\neg z_{i,\$y[i]} \vee \neg z_{j,\$y[j]} \vee \$y[i] > \$y[j] \right) \\
& \qquad\ \ \wedge \left. 
                   \bigwedge_{\$y[i]\in y[i]}\bigwedge_{\substack{j\in Lmin[i]\\\$y[j]\in y[j]}} \left(\neg z_{i,\$y[i]} \vee \neg z_{j,\$y[j]} \vee \$y[i] < \$y[j] \right)\right) 
	\label{eqMapping}
  \end{split}
	\end{equation}
	\label{satsound}
\end{theorem}
\normalsize
\begin{proof}
		Let us show that if $x$ op-matches $y$ then $\phi$ is satisfiable, and if $x$ does not op-match $y$ then $\phi$ is not satisfiable. 
	$(\Rightarrow)$ When $x\approx y$, there is an assignment of values to $y$, $\$y$, that satisfy the orderings of $x$. $\phi$ is satisfiable if there is at least one variable assigned to true per clause $\vee_{\$y[i]\in y[i]}$ $z_{i,\$y[i]}$ given conflicts $\neg z_{i,\$y[i]}\vee\neg z_{j,\$y[j]}$. As conflicts do not prevent the existence of a valid assignment (by assumption), then $\exists_{\$y}\wedge_{i\in \{0..m-1\}} z_{i,\$y[i]}$ and $\phi$ is satisfiable. 
	$(\Leftarrow)$ When $x$ does not op-match $y$, there is no assignment of values $\$y\in y$ that can satisfy the orders of $x$. Per formulation, the conflicts $\neg z_{i,\$y[i]}\vee\neg z_{j,\$y[j]}$ prevent the satisfiability of one or more clauses $\vee_{\$y[i]\in y[i]}$ $z_{i,\$y[i]}$, leading to a non-satisfiable formula.
\end{proof}

If the established $\phi$ formula is satisfiable, there is a Boolean assignment to the variables that specify an assignment of characters in $y$, $\$y$, preserving the orders of $x$ (as defined by \textit{Leq}, \textit{Lmax} and \textit{Lmin}). Otherwise, it is not possible to select an assignment $\$y$ op-matching $x$. 
$\phi$ has at most $r\times m$ variables, $\{z_{i,\sigma} \mid i \in \{0..m-1\},\ \sigma\in\Sigma\}$. %
The Boolean value assigned to a variable $z_{i,\sigma}$ simply defines that the associated character $\sigma$ from $y[i]$ can be either considered (when \textsf{true}) or not (when \textsf{false}) to compose a valid assignment $\$y$ that op-matches the given determinate string $x$.
The reduced formula in \eqref{eqMapping} is composed of two major types of clauses: $\vee_{\$y[i]\in y[i]} z_{i,\$y[i]}$, and $(\neg z_{i,\$y[i]}\vee\neg z_{j,\$y[j]} \vee \textsf{bool})$ where \textsf{bool} is either given by $\$y[i] =\$y[j]$, $\$y[i] <\$y[j]$ or $\$y[i] > \$y[j]$. Clauses of the first type specify the need to select at least one character per position in $y$ to guarantee the presence of valid assignments. %
The remaining clauses specify ordering constraints between characters. If an inequality, such as $\$y[i] > \$y[j]$, is assessed as \textsf{true}, the associated clause is removed. Otherwise, $(\neg z_{i,\sigma_1}\vee\neg z_{j,\sigma_2})$ is derived, meaning that these $\sigma_1$ and $\sigma_2$ characters should not be selected simultaneously since they do not satisfy the orders defined by a given pattern. For instance, the pairs of characters in orange from Figure~\ref{ordersFig} should not be simultaneously selected due to order conflicts. To this end, $(\neg z_{0,2}\vee \neg z_{3,1})$ and $(\neg z_{1,4}\vee \neg z_{2,5})$ clauses need to be included to verify if $y\approx x$. 
Considering $y=(2,4|5,4|5,1|2)$ and $x=(1,4,3,1)$, schematically represented in Figure~\ref{ordersFig}, the associated CNF formula is:
\[
\phi = z_{0,2} \wedge (z_{1,4} \vee z_{1,5}) \wedge (z_{2,4} \vee z_{2,5}) \wedge (z_{3,1} \vee z_{3,2}) \wedge (\neg z_{0,2}\vee \neg z_{3,1} ) \wedge (\neg z_{1,4}\vee \neg z_{2,5} )
\]

%
\begin{theorem}
	Given two strings of length $m$, one being indeterminate with $r=2$, the $\mu$OPPM problem can be reduced to a 2SAT problem with a CNF formula with $O(m)$ size.
	\label{2satsize}
\end{theorem}

\begin{proof}
	Given Theorem~\ref{satsound} and the fact that the reduced CNF formula has at most two literals per clause -- $\phi$ is a composition of $\vee_{\$y[i]\in y[i]} z_{i,\$y[i]}$ clauses with $|y[i]|\in\{1,2\}$ and $(\neg z_{i,\$y[i]}\vee\neg z_{j,\$y[j]} \vee \textsf{bool})$ clauses -- $\mu$OPPM with $r=2$ and one indeterminate string is reducible to 2SAT. %
	The reduced formula has at most $10m$ clauses with 2 literals each, being linear in $m$:
	\begin{itemize}
		\item $ $[\textit{clauses that impose the selection of at least one character per position in $y$}] Since $y$ has $m$ positions, and each position is either determinate (unitary clause) or defines an uncertainty between a pair of characters, there are $m$ clauses and at most $2m$ literals;
		\item $ $[\textit{clauses that define the ordering restrictions between two variables}] A position in the indeterminate string $y[i]$ needs to satisfy at most two order relations. Considering that $i$, $Leq[i]$, $Lmax[i]$ and $Lmin[i]$ specify uncertainties between pairs of characters, there are up to 12 restrictions per position: 4 ordering restrictions between  characters in $y[i]$ and $y[Leq[i]]$, $y[Lmax[i]]$ and $y[Lmin[i]]$. Whenever the order between two characters is not satisfied, a clause is added per position, leading to at most $12m$ clauses.\vskip -0.5cm
	\end{itemize}
\end{proof}

\begin{theorem}
	The $\mu$OPPM between determinate and indeterminate strings of equal length can be solved in linear time when $r=2$.
	\label{time1}
\end{theorem}

\begin{proof}
	Given the fact that a 2SAT problem can be solved in linear time \cite{aspvall1979linear}\footnote{2SAT problems have linear time and space solutions on the size of the input formula. Consider for instance the original proposal \cite{aspvall1979linear}, the formula $\phi$ is modeled by a directed graph $G=(V,E)$, with two nodes per variable $z_i$ in $\phi$ ($z_i$ and $\neg z_i$) and two directed edges for each clause $z_i\vee z_j$ (the equivalent implicative forms $\neg z_i\Rightarrow z_j$ and $\neg z_j\Rightarrow z_i$). Given $G$, the strongly connected components (SCCs) of $G$ can be discovered in $O(|V|+|E|)$. During the traversal if a variable and its complement belong to the same SCC, then the procedure stops as $\phi$ is determined to be unsatisfiable. Given the fact that both $|V|=O(m)$ and $|E|=O(m)$ by Lemma~\ref{2satsize}, this procedure is $O(m)$ time and space.}, this proof directly derives from Theorem~\ref{2satsize} as it guarantees the soundness of reducing $\mu$OPPM ($r=2$) to a 2SAT problem with a CNF formula with $O(m)$ size. 
\end{proof}

As the size of the mapped CNF formula $\phi$ is $O(m)$ and the a valid algorithm to verify its satisfiability would require the construction of a graph with $O(m)$ nodes and edges, the required memory for the target $\mu$OPPM problem is $\Theta(m)$.

%
%

When moving from one to two indeterminate strings, previous contributions are insufficient to answer the $\mu$OPPM problem. 
In this context, the \textit{Leq}, \textit{Lmax} and \textit{Lmin} vectors need to be redefined to be inferred from an indeterminate string:
%

\begin{definition}For $i=0,\ldots,m-1$:
	\begin{itemize}
		\item $\textit{Leq}_x[i|j] =\{k \mid k<i,\ \exists_p\ \$x_j[i]=\$x_p[k]\}$ ($\emptyset$ if there is no eligible $k$),
		\item $\textit{Lmax}_x[i|j]=\{k \mid k<i,\ \exists_p\ \$x_j[i]>\$x_p[k]\}$ ($\emptyset$ if there is no eligible $k$),
		\item $\textit{Lmin}_x[i|j]=\{k \mid k<i,\ \exists_p\ \$x_j[i]<\$x_p[k]\}$ ($\emptyset$ if there is no eligible $k$).
	\end{itemize} 
\label{lminmaxb}
\end{definition}

Figure~\ref{ordersFig2} schematically represents the order relationships of $x=(2,1|3,3)$ and the associated \textit{Leq}, \textit{Lmax} and \textit{Lmin} vectors. In this scenario, %
$x[2]$ needs to be verified not only against $x_0[1]$ but also against $x_1[1]$ in case $x_0[1]$ is disregarded. 

\begin{figure}[t]
	\centering
	\begin{minipage}[b]{\linewidth}\centering
		\begin{tikzpicture}[xscale=1.1,yscale=1.1]
		\SetGraphUnit{1.3}
		\SetVertexNormal[LineColor=white,LineWidth=2pt]
		\Vertex[L={$x[0]$}]{x1}\EA[L={$x_0[1]$}](x1){x2}\EA[L={$x_1[1]$}](x2){x3}\EA[L={$x[2]$}](x3){x4}%
		\SetVertexNormal[LineColor=black,LineWidth=.8pt]
		\SO[L={2}](x1){t1}
		\SO[L={1}](x2){t2}
		\SO[L={3}](x3){t3}
		\SO[L={3}](x4){t4}
		\tikzset{EdgeStyle/.append style={->, bend right = 45, color=purple}}
		\Edge[label=$<$](t4)(t1)
		\tikzset{EdgeStyle/.append style={->, bend right = 30, color=purple}}
		\Edge(t3)(t1)
		\Edge(t4)(t2)
				\tikzset{EdgeStyle/.append style={->, bend left = 50, color=green}}
		\Edge(t4)(t3)
		\tikzset{EdgeStyle/.append style={->, bend left = 50, color=blue}}
		\Edge[label=$>$](t2)(t1)
		\end{tikzpicture}	
	\end{minipage}
\\
\begin{minipage}[t]{\linewidth}\centering
		\begin{tabular}{lm{.35cm}m{.35cm}m{.35cm}m{0.70cm}}\toprule
		Pattern&2&1&3&3\\
		$i$ &0&1&1&2\\
		$j$&0&0&1&0\\\midrule
		\mbox{\textit{Leq}$[i|j]$}&$\emptyset$&$\emptyset$&$\emptyset$&\{1\}\\
		Ordered indexes (asc)&1&0&2&3\\
		\mbox{\textit{Lmax}$[i|j]$}&$\emptyset$&$\emptyset$&\{0\}&\{0,1\}\\
		Ordered indexes (desc)&2&3&0&1\\
		\mbox{\textit{Lmin}$[i|j]$}&$\emptyset$&\{0\}&$\emptyset$&$\emptyset$\\\bottomrule
	\end{tabular}
\end{minipage}%
	\caption{Order relationships of $x=(2,1|3,3)$ and the corresponding \textit{Leq}, \textit{Lmax} and \textit{Lmin} vectors.}
	\label{ordersFig2}
\end{figure}
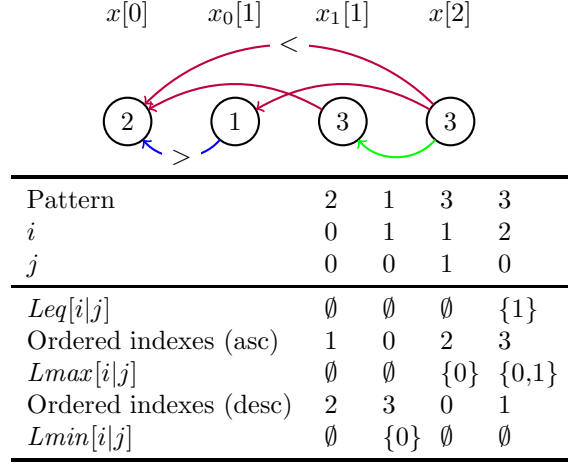

\begin{remark} Given \textit{Leq}, \textit{Lmax} and \textit{Lmin} (Definition~\ref{lminmaxb}), there are $O((rm)^2)$ order relationships when $r\in\mathbb{N}^+$ since each character in a given position establishes at most $O(m)$ relationships with characters in preceding positions. %
	\label{cor1}
\end{remark} 

\begin{lemma}
	\label{orderlemma2}
Given indeterminate strings $x$ and $y$, let $A_j=\textit{Leq}_x[t+1|j]$, $B_j=Lmax_x[t+1|j]$ and $C_j=Lmin_x[t+1|j]$ (Definition~\ref{lminmaxb}) be the orders associated with $\$x_j[t+1]$. If $x[1 . . t]\approx y[1 . . t]$ is verified on a partial assignment of $y$ characters, denoted by $\S y$, then $x[1 . . t + 1]\approx y[1 . . t + 1]$ if and only if
\begin{eqnarray*}
  \exists_{j\in\{0,1\}}\ \exists_{\$y[t+1]\in \S y[t+1]}\ \forall_{a\in A_j, b\in B_j, c\in C_j}
	\ \exists_{\$y[a]\in \S y[a], \$y[b]\in \S y[b], \$y[c]\in \S y[c]} \\
  \left( \$y[t+1]= \$y[a] \wedge \$y[t + 1] > \$y[b] \wedge \$y[t + 1] < \$y[c]\right)
\end{eqnarray*}
\end{lemma}
\begin{proof}
	$(\Rightarrow)$ Similar to the proof of Lemma~\ref{orderlemma}, yet $A$, $B$ and $C$ conditional to $x[t+1]$ (Definition~\ref{minmaxoriginal}) are now given by $A_j$, $B_j$ and $C_j$ conditional to $x_j[t+1]$ (Definition~\ref{lminmaxb}). If there is an assignment to $y[1 . . t + 1]$ in $\S y$ that preserves one of the possible orders in $x[1 . . t + 1]$, then for any $a\in A_j$, $b\in B_j$ and $c\in C_j$: $\$y[t + 1]=\$y[a]$, $\$y[t + 1]>\$y[b]$ and $\$y[t + 1]<\$y[c]$ (where $\$y[t + 1]\in \S y[t+1]$, $\$y[a]\in \S y[a]$, $\$y[b]\in \S y[b]$, and $\$y[c]\in \S y[c]$).

$(\Leftarrow)$ We need to show that $x[1 . . t + 1]\approx y[1 . . t + 1]$. Since $x[1 . . t]\approx y[1 . . t]$, it is sufficient to prove that for $i\le t$: exists $\$x[i]\in \S x[i]$, $\$x[t+1]\in \S x[t+1]$, $\$y[i]\in \S y[i]$, and $\$y[t + 1]\in \S y[t+1]$ such that
$\$x[t+1]=\$x[i] \Leftrightarrow \$y[t+1]=\$y[i]$, $\$x[t+1]>\$x[i] \Leftrightarrow \$y[t+1]>\$y[i]$ and $\$x[t+1]<\$x[i] \Leftrightarrow \$y[t+1]<\$y[i]$. This results from Definition~\ref{lminmaxb}, the order-isomorphism property and Lemma~\ref{orderlemma}.
\end{proof}

\begin{figure}[t]
	\centering
	\begin{tikzpicture}[xscale=1,yscale=.8]
	\SetGraphUnit{1.3}
	\SetVertexNormal[LineColor=white,LineWidth=2pt]
	\Vertex[L={$x[0]$}]{x1}\EA[L={$x_0[1]$}](x1){x2}\EA[L={$x_1[1]$}](x2){x3}\EA[L={$x[2]$}](x3){x4}
	\scriptsize
	\SetVertexNormal[LineColor=orange,LineWidth=1.5pt]
	\SO[L={$y[0]$=2}](x1){t1}
	\SO[L={$y[1]$=0}](x3){t3}
	\SO[L={$y_0[2]$=3}](x4){t4}
	\SO[L={$y_1[2]$=4}](1){t5}
	\SetVertexNormal[LineColor=black,LineWidth=.8pt]
	\SO[L={$y[1]$=0}](x2){t2}
	\tikzset{EdgeStyle/.append style={-, bend right = 40, color=orange, line width=1.5pt}}
	\Edge(t3)(t4)
	\Edge(t3)(t5)
	\tikzset{EdgeStyle/.append style={-, bend right = 50, color=orange, line width=1.5pt}}
	\Edge(t3)(t1)
	\end{tikzpicture}
	\caption{Conflicts when op-matching $y=(2,0,3|4)$ against $x=(2,1|3,3)$.}
	\label{ordersFig3}
\end{figure}
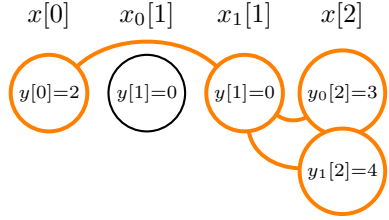

Figure~\ref{ordersFig3} represents encountered restrictions when op-matching $x=(2,1|3,3)$ against $y=(2,0,3|4)$. The right side edges capture the detected incompatibilities, i.e. pairs of characters that cannot be selected simultaneously. For the given example, there are 2 valid assignments -- $\$y_1=(2,0,3)$ and $\$y_2=(2,0,4)$ -- satisfying $\$x_0[1]<\$x_0[0]<\$x_0[2]$, thus $x\approx y$. 

To verify whether there is an assignment that satisfies the identified ordering restrictions, Theorem~\ref{2satsoundb} extends the previously introduced SAT mapping given by \eqref{eqMapping}.

\begin{theorem}
Given \textit{Leq}, \textit{Lmax} and \textit{Lmin} (Definition~\ref{lminmaxb}), $\mu$OPPM problem over two indeterminate strings of equal length can be reduced to a satisfiability problem with the following CNF formula:
\begin{equation}
\begin{split}
&\phi = \bigwedge_{i=0}^{m-1} \bigvee_{\substack{\$y[i]\in y[i]\\\$x[i]\in x[i]}} 
   z_{i,\$x[i],\$y[i]} \\
&\wedge \bigwedge_{i=0}^{m-1} \bigwedge_{\substack{\$y[i]\in y[i]\\\$x[i]\in x[i]}}\left(
   \bigwedge_{j\in Leq[i]}  \bigwedge_{\substack{\$y[j]\in y[j]\\\$x[j]\in x[j]}} \left(\neg z_{i,\$x[i],\$y[i]}\vee \neg z_{j,\$x[j],\$y[j]} \vee \$y[i] = \$y[j] \right) \right.\\
&\qquad\qquad\quad\ \ \wedge 
   \bigwedge_{j\in Lmax[i]} \bigwedge_{\substack{\$y[j]\in y[j]\\\$x[j]\in x[j]}} \left(\neg z_{i,\$x[i],\$y[i]}\vee \neg z_{j,\$x[j],\$y[j]} \vee \$y[i] > \$y[j] \right) \\
&\qquad\qquad\quad\ \ \wedge \left.
   \bigwedge_{j\in Lmin[i]} \bigwedge_{\substack{\$y[j]\in y[j]\\\$x[j]\in x[j]}} \left(\neg z_{i,\$x[i],\$y[i]}\vee \neg z_{j,\$x[j],\$y[j]} \vee \$y[i] < \$y[j] \right)\right)
\label{eqMappingb}
\end{split}
\end{equation}
\label{2satsoundb}
\end{theorem}
\begin{proof}
		If $x\approx y$ then $\phi$ is satisfiable, and if $x$ does not op-match $y$ then $\phi$ is not satisfiable. 

	$(\Rightarrow)$ When $x$ op-matches $y$, there is an assignment of values in $x$ and $y$ such that $\$x\approx\$y$. $\phi$ is satisfiable if there is at least one valid assignment $z_{i,\$x[i],\$y[i]}$ per $i^{\textit{th}}$ position. As conflicts $\neg z_{i,\$x[i],\$y[i]}\vee\neg z_{j,\$x[j],\$y[j]}$ do not prevent the existence of a valid assignment (by assumption), one or more variables $z_{i,\$x[i],\$y[i]}$ can be selected per position. $\phi$ can then be satisfied by fixing a single variable $z_{i,\$x[i],\$y[i]}$ per $i^{\textit{th}}$ position as \textsf{true} and the remaining variables as \textsf{false}. %
	$(\Leftarrow)$ When $x$ does not op-match $y$, there is no assignment of values $\$x\in x$ and $\$y\in y$ such that $\$x\approx\$y$. %
	Per formulation, in the absence of an order-preserving match, conflicts will prevent the assignment of at least one variable $z_{i,\$x[i],\$y[i]}$ per $i^{\textit{th}}$ position, thus making $\phi$ formula unsat. %
	\end{proof}

If the formula in \eqref{eqMappingb} is satisfiable, there is a Boolean assignment to the variables such that there is an assignment of characters in $y$, $\$y$, and in $x$, $\$x$, such that both strings op-match. Otherwise, it is not possible to select assignments such that $x\approx y$. 
Given $r=2$, the established $\phi$ formula has at most %
$4m$ variables, $\{z_{i,\sigma_1,\sigma_2} \mid i \in \{0\ldots m-1\}$, $\sigma_1,\sigma_2\in\Sigma$\}. %
The Boolean values assigned to these variables %
define whether characters $\sigma_1\in x[i]$ and $\sigma_2\in y[i]$ belong to an op-match. %
The reduced formula is composed of two major types of clauses:
\begin{itemize}
\item Those in the first line of \eqref{eqMappingb} ensure that at least one combination of characters, $\$x[i]$ and $\$y[i]$, should be selected per $i^{\textit{th}}$ position.
\item Remaining ones in \eqref{eqMappingb} specify ordering constraints between pairs of characters $\sigma_1\in y[i]$ and $y[Leq[i]]$, $y[Lmax[i]]$ and $y[Lmin[i]]$; if the inequalities $\$y[i]=\$y[j]$, $\$y[i]>\$y[j]$ and $\$y[i]<\$y[j]$ are assessed as \textsf{false}, then it leads to clauses of the form $(\neg z_{i,\sigma_1}\vee\neg z_{j,\sigma_2})$, meaning that these characters should not be selected simultaneously in the given positions (see Figure~\ref{ordersFig3}). %
\end{itemize}

To instantiate the proposed mapping, consider $x=(2,1|3,3)$ and $y=(2,0,3|4)$, schematically represented in Figure~\ref{ordersFig2}. The associated CNF formula is:
\begin{equation*}
\begin{split}
\phi & =  z_{0,2,2}\wedge(z_{1,1,0} \vee z_{1,3,0})\wedge (z_{2,3,3} \vee z_{2,3,4}) \\ 
& \wedge (\neg z_{0,2,0} \vee \neg z_{1,3,0}) \wedge (\neg z_{1,3,0} \vee \neg z_{2,3,3})\wedge (\neg z_{1,3,0} \vee \neg z_{2,3,4}) 
\end{split}
\end{equation*}

\begin{theorem}
	The $\mu$OPPM problem for two indeterminate strings of equal length is reducible into a satisfiability problem over a CNF formula with size $O((mr)^2)$.%
	\label{2satsizeb}
\end{theorem}

\begin{proof}
	The reduced formula in \eqref{eqMappingb} is in the two conjunctive normal form (CNF) with at most $4m$ clauses in the first line of \eqref{eqMappingb} and a maximum of $O(mr)$ orders per position (Remark~\ref{cor1}), totalling at most $O((mr)^2)$ order conflicts between characters, from the restriction clauses in the reammining of \eqref{eqMappingb}. %
\end{proof}

Although we are no longer in the conditions of Theorem~\ref{time1}, namely because the above satisfiability formulation is not a 2SAT instance, given its unique properties, effective backtracking in accordance with the clauses in the first line of \eqref{eqMappingb}, as well as dedicated conflict pruning principles derived from reamining clauses in \eqref{eqMappingb}, can be considered to develop efficient SAT solvers able to solve the $\mu$OPPM problem.
And, as we will show later, we are not expected to do much better.

 %

\subsection{Polynomial time Alternate-$\mu$OPPM}\label{sec:alternateoppm}
In this section, we define Alternate-$\mu$OPPM as the subproblem of $\mu$OPPM where both strings ($x$ and $y$, interchangeable) may have indeterminate characters, but never in the same position; we show that Alternate-$\mu$OPPM is polynomial in both the number of indeterminacies ($r$, which may be different in each position and string) and length of the strings ($m$). To do this, we will present a set of 2SAT clauses, in the form of implications, that can represent every constraint of this problem.
We will first assume that there are no repeated characters within each string and then extend the reduction to handle equalities.

Given a string $x$ and position $i$, we represent the set of indeterminate characters $x[i]$ as the ascending sequence $a_0 \vert ... \vert a_{{r_i}-1}$ where $\forall_j\ a_j\in x[i]$ and $\vert x[i] \vert = r_i $. We will use only $r$ when the context leads to no ambiguities, or to mean the largest possible $r_i$. All of our 2SAT variables will be of the form $g_{a_j}$, meaning that the chosen value $\$x[i]$ is greater than or equal to $a_j$.

\paragraph{Consistency clauses} Here, we describe the clauses that maintain consistency between all the $g$ variables for \textit{individual} positions. We only need to specify that, if we have chosen a value greater than $a_i$, we have also chosen a value greater than $a_{i-1}$, the value immediately below it, i.e.,
\begin{equation*}
\forall_{i \in [1,r-1]} (g_{a_i} \implies g_{a_{i-1}}).
\end{equation*}
This leads to a single clause per indeterminacy, per position, for both pattern and text, and so, at most, $2 m r = O(m r)$ clauses.

\paragraph{Order clauses (Type $1$)} Here, we describe the clauses enforcing the order relation between \textit{each pair} of positions. Given two strings $x$ and $y$, for positions $\alpha$ and $\beta$, if $\$x[\alpha] > \$x[\beta]$, then $\$y[\alpha] > \$y[\beta]$ (and the same for the $<$ relation).

This first set of clauses applies to Type $1$ (see Table~\ref{table:pair_type1}). We only need to find the index (in each string) that separates the cases where $\$x[\alpha] > \$x[\beta]$ from the cases where $\$x[\alpha] < \$x[\beta]$ and add a single constraint expressing it.

\begin{table}[t]
    \center
    \caption{ Type $1$ of pairs we can have in Alternate-$\mu$OPPM. \label{table:pair_type1}}
    \begin{tabular}{|c|c|c|}
    \hline
    $i$ & $\alpha$     & $\beta$ \\ \hline
    $x$ & $a$    & $b_0 \vert ... \vert b_{r_{\beta}-1}$ \\ \hline
    $y$ & $a_0 \vert ... \vert a_{r_{\alpha}-1}$     & $b$  \\ \hline
    \end{tabular}
\end{table}

Let $i$ be the lowest index such that $b_i > a$ and $j$ the lowest index such that $a_j \geq b$, where $a$ and $b$ are as in Table~\ref{table:pair_type1}. Then, we have
\begin{equation*}
\begin{split}
g_{b_i} \implies \lnot g_{a_j},\\
\lnot g_{b_i} \implies g_{a_j}.
\end{split}
\end{equation*}
This leads to two clauses for every pair of positions, and so, $O(m^2)$ clauses.

\paragraph{Order clauses (Type $2$)} 
Finally, we have a second set of clauses that applies to Type $2$ (see Table~\ref{table:pair_type2}). Here, we have the order between $\alpha$ and $\beta$ fixed already by whichever string $x$ or $y$ has no indeterminacies.

\begin{table}[t]
    \center
    \caption{ Type $2$ of pairs we can have in Alternate-$\mu$OPPM. \label{table:pair_type2}}
    \begin{tabular}{|c|c|c|}
    \hline
    $i$ & $\alpha$     & $\beta$ \\ \hline
    $x$ & $a$    &  $b$ \\ \hline
    $y$ & $a_0 \vert ... \vert a_{r_{\alpha}-1}$     & $b_0 \vert ... \vert b_{r_{\beta}-1}$  \\ \hline
    \end{tabular}
\end{table}

If $a>b$, for every index $i$ indexing $b_i$, and let $j$ be the lowest index such that $a_j > b_i$. Then we add
\begin{equation*}
g_{b_i} \implies g_{a_j}.
\end{equation*}
If there is no such $j$, we add instead
\begin{equation*}
\lnot g_{b_i}.
\end{equation*}

Similarly, if $a<b$, for every index $i$ indexing $a_i$, let $j$ be the lowest index such that $a_i<b_j$. Then we add
\begin{equation*}
g_{a_i} \implies g_{b_j}.
\end{equation*}
If there is no such $j$, we add instead
\begin{equation*}
\lnot g_{a_i}.
\end{equation*}

This leads to at most $r$ clauses for every pair of positions, and so $O(r m^2)$ clauses. Because character order is a transitive property, this type of clauses may be reduced to $O(r m)$ using a similar notion to the \textit{Lmax} and \textit{Lmin} sets introduced in Section~\ref{2indsec} to consider only ``adjacent'' (taking adjacent to mean the closest position of the same type) pairs of positions, instead of every pair.

\paragraph{Forcing choice}
With the clauses specified above, we can find coherent solutions to the problem. However, it is possible to satisfy the formula by assigning all possible values for a given variable to false (effectively skipping the position). This has a straightforward solution, given the chosen encoding of the variables. Each 2SAT variable represents a \textit{greater or equal} value in the corresponding OPPM position, the variable corresponding to the lowest value for each position is trivially true, letting us force a value choice with a single added variable. For every position, with variables $g_0,...g_{r_i}$, we add the clause $g_0$, forcing it to be true to satisfy the 2SAT formula.

\paragraph{Extracting solutions}
Finally, we need to extract the solution to the OPPM problem from the 2SAT solution. This is easily done in linear time by sweeping every variable in ascending order, in each position. In each position, with variables $g_0,...,g_{r_i}$, we find the variable at index $j$ such that $g_j$ is true and $g_{j+1}$ is false. The chosen value in the OPPM problem, for the given position, is the value at index $j$.

\paragraph{Dealing with equalities}\label{sec:alternateequal}
We now turn to cases where characters match and show how to adapt the encoding above to equalities.
Let us consider Type II equalities, first, where $a=b$. The easy solution to this is the same as the one presented before. We preprocess the two strings by grouping all the repeats into a single position and intersecting their indeterminacies.
For Type I equalities, we need to add $4$ clauses to each pair. Let $i,j$ be indexes such that $a=b_i$ and $b=a_j$. We add
\begin{equation*}
\begin{split}
g_{b_i} & \implies g_{a_j},\\
g_{a_j} & \implies g_{b_i},\\
\lnot g_{b_{i+1}} & \implies \lnot g_{a_{j+1}},\\
\lnot g_{a_{j+1}} & \implies \lnot g_{b_{i+1}}.\\
\end{split}
\end{equation*}
If only $i$ exists (or $j$), we simply remove $b_i$ (or $a_j$) from the input, as such an assignment could never lead to a valid solution.

\paragraph{Pair incompatibility}
All the clauses described above serve to maintain consistency between pairs. It may happen that a given pair is unsatisfiable by itself, and no clauses would be constructed. These cases can be dealt separately, as pre-processing. If we find a pair that can not be satisfied, we can terminate the program before ending the construction, since there is no solution to the OPPM instance.

\begin{theorem}\label{time-alt}
	The Alternate-$\mu$OPPM can be solved in $O(r m^2)$ time and space.
\end{theorem}
\begin{proof}Property resulting from the encoding above and, as in the proof of Theorem~\ref{time1}, given the fact that a 2SAT problem can be solved in linear time~\cite{aspvall1979linear}.
\end{proof}

\subsection{$\mu$OPPM with 3 indeterminacies in both text and pattern is \textbf{NP}-hard}\label{sec:reduction}
In this section, we define $\{3,3\}$-$\mu$OPPM as the subproblem of $\mu$OPPM where both the pattern and the text have indeterminate characters in any position (although at least one position must have at least three indeterminate characters in both pattern and text) and prove it \textbf{NP}-hard (thus proving the same for general $\mu$OPPM). We do this with a direct reduction from 3CNF-SAT, first presenting the construction and then the proof of equivalence between the two instances. The construction is similar to the one by Bose et al. for the permutation matching problem~\cite{permutation}.

\paragraph{Construction}
To ease the description of the construction itself, we start by describing how we represent an instance of 3CNF-SAT. First, we assume that every literal and clause has some ordering. We have a set $V$ of literals, and a set $C$ of clauses. Each clause $c$ is represented by two tuples, $(z_{c,0}, z_{c, 1}, z_{c, 2})$ and $(l_{c,0}, l_{c, 1}, l_{c, 2})$. $z_{c,i} \in \{0, \ldots,\vert V \vert -1\}$ represents the index of literal $i$ of clause $c$; $l_{c,i} \in \{0, 1\}$ represents the value of the literal $i$ in clause $c$, having the value of $0$ for positive literals and $1$ for negative literals. For example, the clause $(v_1 \lor \lnot v_2, \lor v_5)$ would be represented by the two tuples $z=(1, 2, 5)$ and $l=(0, 1, 0)$.

Although the designations of text or pattern are interchangeable in this section, we will use pattern for the simpler string (with less indeterminacies) and text for the more complicated string (with more indeterminacies). We use $p$ and $t$ for the pattern and text, respectively, or $s$ when they are interchangeable.

Both text and pattern have two parts, one representing literals and the other representing clauses. Each literal, and clause, has a single position in each string to represent it, dividing $s$ into $s_V = s[0..\vert V \vert - 1]$ and $s_C = s[\vert V \vert .. \vert V \vert+\vert C \vert - 1]$. In $p_V$, we have a simple sequence of literals given by their indexes, so $p[i] = i+1$, for $i \in \{0, \ldots,\vert V \vert - 1\}$; in $t_V$ we have a similar sequence, but each literal takes one of two variable values to represent an assignment of \textsf{true} or \textsf{false}, so $t[i] =  2\times(i+1)$ or $2\times(i+1)-1$. We choose the larger value to represent the assignment of \textsf{true}. In $s_C$, each position has three indeterminacies, corresponding to the three variables of the clause. In $p_C$, we choose one of the three literals of the respective clause. For clause $c$, with literals $v_1, v_2, v_5$ (regardless of their value being positive or negative), its position in $p$, $p[\vert V \vert + c] = 1 \vert 2 \vert 5$. In $t_C$, as in $p_C$ we choose
one of the literals, but now the \textit{value} of the literal must satisfy the clause. For clause $c$, $(v_1 \lor \lnot v_2 \lor v_5)$, $t[\vert V \vert + c] = 2\times 1 - 0 \mid 2\times 2 - 1 \mid 2\times 5 - 0$ $= 2\vert3\vert10$.
An example of this construction is shown in Table \ref{table:reduction_example}.
\begin{table}[]
    \center
    \caption{ $\mu$OPPM instance corresponding to the 3CNF-SAT formula $(z_1 \lor \lnot z_2 \lor z_3) \land (\lnot z_1 \lor z_2 \lor z_4)$. \label{table:reduction_example}}
    \begin{tabular}{|c|c|c|c|c|c|c|}
    \hline
    $i$ & $0$     & $1$     & $2$     & $3$     & $4$         & $5$         \\ \hline
    Formula & $z_1$    & $z_2$    & $z_3$    & $z_4$    & $c_1$        & $c_2$        \\ \hline
    Pattern & $1$     & $2$     & $3$     & $4$     & $1 \vert 2 \vert 3$ & $1 \vert 2 \vert 4$ \\ \hline
    Text & $1 \vert 2$ & $3 \vert 4$ & $5 \vert 6$ & $7 \vert 8$ & $2 \vert 3 \vert 6$ & $1 \vert 4 \vert 8$ \\ \hline
    \end{tabular}
\end{table}

\begin{lemma}\label{lemma:polyconstr} The construction above takes polynomial time. \end{lemma}
\begin{proof}
It is easy to see that, assuming that variables and clauses are numbered, we can simply scan the formula once to construct our two strings in linear time.
\end{proof}

\begin{lemma}\label{lemma:soleq} The initial 3CNF-SAT clause is satisfiable if and only if there is an order-isomorphic match between the two constructed strings. \end{lemma}
\begin{proof} We start by showing how solving the $\mu$OPPM instance solves the initial 3CNF-SAT instance. To solve $\mu$OPPM, we need to choose exactly one value for each position in $p$ and $t$ that leads to two order-isomorphic strings. To extract the solution, we can limit ourselves to look at the initial part of $t$, $t[0, \vert V \vert-1]$, which sets the value of each literal.

First, note that $p$ function is to maintain consistency between the values of literals chosen in $t$. By choosing only literals in $p$, and not their values, we force equality between all such literals. Because of order-isomorphism, this equality must be kept in $t$, forcing a valid solution to use a single value for each literal (since different values match in $p$ but mismatch in $t$). If we choose a literal to be positive/negative at some position in $t$, we force the value of that literal to be positive/negative at every position in $t$.  

Now, we focus on $t_C$. Every clause has exactly one position in $t_C$, and each of these positions have three choices of value, matching only the three values that satisfy a clause. Because we must choose one value in each position to solve our $\mu$OPPM instance, we must choose one value that satisfies each clause, for every clause.

Putting these two properties together, to solve $\mu$OPPM we must choose a literal value that satisfies each clause \textit{and} those literals must have consistent values. This establishes the equivalence between the solutions of the two instances.

We can easily extract the solution from $\mu$OPPM to 3CNF-SAT by checking whether the values in $t_V$ are even or odd, \textsf{true} or \textsf{false}, respectively. There is a unique solution to 3CNF-SAT given an $\mu$OPPM solution.

To extract the solution from 3CNF-SAT to $\mu$OPPM, we take the values assigned to each variable and choose the respective values in $t_V$. Then, we need to choose values for $p_C$ and $t_C$, which can easily be done by choosing any of the literals that satisfies its respective clause. There may be multiple $\mu$OPPM solutions for a given 3CNF-SAT solution.
\end{proof}

\begin{theorem} \{3,3\}-$\mu$OPPM is \textbf{NP}-hard.
\end{theorem}
\begin{proof} Using Lemmas~\ref{lemma:polyconstr} and~\ref{lemma:soleq} we show that 3CNF-SAT $\leq _p$ $\{3,3\}$-$\mu$OPPM by constructing an instance of $\mu$OPPM in polynomial time. The solutions can also be retrieved and translated in polynomial time.
\end{proof}

\begin{theorem} $\mu$OPPM is \textbf{NP}-hard.\end{theorem}
\begin{proof} Since $\{3,3\}$-$\mu$OPPM is a particular case of $\mu$OPPM, and it is \textbf{NP}-hard, then OPPM is \textbf{NP}-hard.
\end{proof}

\section{\mbox{Polynomial time $\mu$OPPM}}
\vskip -0.9cm\textcolor{white}{.}

\begin{lemma}
	Given a pattern string of length $m$ and a text string of length $n$, one being indeterminate, the $\mu$OPPM problem can be solved in $O(nmr\lg r)$ time.
	\label{uOPPMbound}
\end{lemma}
\begin{proof}
	From Theorem~\ref{timeLIS}, %
	verifying if two strings of length $m$ op-match can be done in $O(mr\lg r)$ time (indetermination in one string) %
	since at most $n-m+1$ verifications need to be performed.
\end{proof}
\vskip -0.1cm

Lemma~\ref{uOPPMbound} confirms that the $\mu$OPPM problem with one indeterminate strings %
is in class \textbf{P}.
This  lemma further triggers the research question ``\textit{Is $O(nmr\lg r)$ %
	a tight bound to solve the $\mu$OPPM?}'', here left as an open research question. 

Irrespectively of the answer, the analysis of the average complexity is of complementary relevance. State-of-the-art research on the exact OPPM problem shows that the average performance of algorithms in $O(nm)$ time can outperform linear time algorithms \cite{Chhabra,cantone2015efficient,chhabra2015alternative}. 

Motivated by the evidence gathered by these works, we suggest the use of filtration procedures to improve the average complexity of the proposed $\mu$OPPM algorithm while still preserving its complexity bounds. 
A filtration procedure encodes the input pattern and text, and relies on this encoding to efficiently find positions in the text with a high likelihood to op-match a given pattern. Despite the diversity of string encodings, simplistic binary encodings are considered to be the state-of-the-art in OPPM research \cite{Chhabra,cantone2015efficient}. In accordance with Chhabra et al. \cite{Chhabra}, a pattern $p$ can be mapped into a binary string $p'$ expressing increases (1), equalities (0) and decreases (0) between subsequent positions. By searching for exact pattern matches of $p'$ in an analogously transformed text string $t'$, we guarantee that the verification of whether $p[0..m-1]$ and $t[i..i+m-1]$ orders are preserved is only performed when exact binary matches occur. Illustrating, given $p=(3,1,2,4)$ and $t=(2,4,3,5,7,1,4,8)$, then $p'=(1,0,1,1)$ and $t'=(1,1,0,1,1,0,1,1)$, revealing two matches $t'[1..4]$ and $t'[4..7]$: one spurious match $t[1..4]$ and one true match $t[4..7]$.

When handling indeterminate strings the concept of increase, equality and decrease needs to be redefined. Given an indeterminate string $x$, consider
$x'[i]=1$ if $max(x[i])<min(x[i+1])$, $x'[i]=0$ if $min(x[i])\ge max(x[i+1])$, and $x'[i]=\ast$ otherwise.
Under this encoding, the pattern matching problem is identical under the additional guard that a character in $p'$ always matches a do not care position, $t'[i]=\ast$, and vice-versa. Illustrating, given $p=(6,2|3,5)$ and $t=(3|4,5,6|8,6|7,3,5,4|6,7|8,4)$, then $p'=(0,1)$ and $t'=(11\ast01\ast10)$, leading to one true match $t[3..5]$ -- e.g. $\$t[3..5]=(6,3,5)$ -- and one spurious match $t[5..7]$. 
Exact pattern matching algorithms, such as Knuth-Morris-Pratt and Boyer-Moore, can be adapted to consider do not care positions while preserving complexity bounds~\cite{kpm,boyermoore}. 

The properties of the proposed encoding guarantee that the exact matches of $p'$ in $t'$ cannot skip any op-match of $p$ in $t$. Thus, when combining the premises of Lemma~\ref{uOPPMbound} with the previous observation, we guarantee that the computed $\mu$OPPM solution is sound.

The application of this simple filtration procedure prevents the recurring $O(mr\lg r)$ %
verifications $n-m+1$ times. Instead, the complexity of the proposed method to solve the $\mu$OPPM problem becomes $O(dmr\lg r+n)$ (when one string is indeterminate) %
where $d$ is the number of exact matches ($d\ll n$). %
According to previous work on exact OPPM with filtration procedures \cite{Chhabra}, SBNDM2 and SBNDM4 algorithms \cite{vdurian2010improving} (Boyer-Moore variants) %
were suggested to match binary encodings. In the presence of small patterns, Fast Shift-Or (FSO) \cite{fredriksson2005practical} can be alternatively applied \cite{Chhabra}.

A given string text can be read and encoded incrementally from the standard input as needed to perform $\mu$OPPM, thus requiring $O(mr)$ space. When filtration procedures are considered, the aforementioned algorithms for exact pattern matching require $O(m)$ space \cite{Chhabra}, thus $\mu$OPPM space requirements are bound by substring verifications (Section~\ref{secSol1}): $O(mr)$ space when one string is indeterminate and $O((mr)^2)$ when indetermination is considered on both strings. %

\section{Open problem}\label{sec:open_problem}
We can look at the $\mu$OPPM by the number and position of the indeterminate characters. We have shown that, for any number of indeterminacies, $\mu$OPPM has a polynomial-time algorithm for indeterminate characters in a single string (Section \ref{uOPPM1}), or in both strings, but never in both strings at the same position (Section \ref{sec:alternateoppm}). For indeterminate characters in both strings at the same position, we have also shown that for \textit{at least three} indeterminacies (at select positions), the problem in \textbf{NP}-hard (Section \ref{sec:reduction}). 

There is a gap in between these two groups, however, for the strings where there are \textit{at most two} indeterminate characters in both strings at the same position. It remains open whether or not this problem is \textbf{NP}-hard. Given that our reduction from Section \ref{sec:reduction} uses three indeterminate character in both strings, it also remains open whether the problem with two indeterminate characters in one string and three in the other (at the same position) is \textbf{NP}-hard.

Following the pattern-avoidance precedent by Guillemot and Vialette~\cite{permutation_avoid} for the related problem of permutation matching, we note that, for the case of $\mu$OPPM with \textit{at most two} indeterminate characters (both strings, same position), there is a straightforward encoding in 2SAT for $(1\vert3, 2\vert4)$-avoiding strings, here taken to mean that, in a single string, for the pair of positions $(i, j)$, the rank of the characters (only for the pair in question) is not $1\vert3$ in $i$ and $2\vert4$ in $j$ (with $i$ and let $j$ being interchangeable). The full problem, however, remains open.

\section{Concluding remark}
This work addressed the relevant yet scarcely studied problem of finding order-preserving pattern matches on indeterminate strings ($\mu$OPPM). We showed that the problem has a linear time and space solution when one string is indeterminate.
In addition, the $\mu$OPPM problem (when both strings are indeterminate) was mapped into a satisfiability formula of polynomial size and two simple types of clauses in order to study efficient solvers for the $\mu$OPPM problem.
Moreover the $\mu$OPPM problem was shown to be \textbf{NP}-hard in general.
Finally, we showed that solvers of the $\mu$OPPM problem can be boosted %
in the presence of filtration procedures and we identified a still open problem in what concerns
the computational complexity of the $\mu$OPPM problem when restricted to at most two indeterminate
characters in both strings at the same position.

\paragraph{Acknowledgments}
This work was developed in the context of a secondment granted by the BIRDS MASC RISE project funded in part by EU H2020 research and innovation programme
under the Marie Skłodowska-Curie grant agreement no.\textsf{690941}. This work was further supported by national funds through Fundação para a Ciência e Tecnologia (FCT), namely under projects \textsf{PTDC/CCI-BIO/29676/2017}, \textsf{TUBITAK/0004/2014}, \textsf{SAICTPAC/0021/2015}, and \textsf{UID/CEC/50021/2019}.
 %

\bibliography{references}
\end{document}